\documentclass{IEEEtran}
\usepackage{soul,times,amssymb,amsthm,amsmath,amsfonts,bbm,mathrsfs,xcolor,subfigure,graphicx,enumitem,booktabs,euscript,nicefrac}
\allowdisplaybreaks
\usepackage{graphicx}
\usepackage[unicode]{hyperref}
\hypersetup{
    colorlinks,
    linkcolor={blue!80!black},
    citecolor={blue!80!black},
    urlcolor={blue!80!black}
}
\usepackage[normalem]{ulem}
\usepackage[utf8]{inputenc}
\usepackage[T1]{fontenc}

\newtheorem{theorem}{Theorem}[section]
\newtheorem{lemma}[theorem]{Lemma}

\newtheorem{proposition}[theorem]{Proposition}
\newtheorem{corollary}[theorem]{Corollary}
\theoremstyle{remark}
\newtheorem{remark}{Remark}
\newtheorem{definition}{Definition}[section]

\newcommand\nc\newcommand
\nc\bfa{{\boldsymbol a}}\nc\bfA{{\boldsymbol A}}\nc\cA{{\EuScript A}}
\nc\bfb{{\boldsymbol b}}\nc\bfB{{\boldsymbol B}}\nc\cB{{\EuScript B}}
\nc\bfc{{\boldsymbol c}}\nc\bfC{{\boldsymbol C}}\nc\cC{{\mathscr C}}
\nc\bfd{{\boldsymbol d}}\nc\bfD{{\boldsymbol D}}\nc\cD{{\EuScript D}}
\nc\bfe{{\boldsymbol e}}\nc\bfE{{\boldsymbol E}}\nc\cE{{\EuScript E}}
\nc\bff{{\boldsymbol f}}\nc\bfF{{\boldsymbol F}}\nc\cF{{\mathscr F}}
\nc\bfg{{\boldsymbol g}}\nc\bfG{{\boldsymbol G}}\nc\cG{{\EuScript G}}
\nc\bfh{{\boldsymbol h}}\nc\bfH{{\boldsymbol H}}\nc\cH{{\mathcal H}}
\nc\bfi{{\boldsymbol i}}\nc\bfI{{\boldsymbol I}}\nc\cI{{\mathcal I}}
\nc\bfj{{\boldsymbol j}}\nc\bfJ{{\boldsymbol J}}\nc\cJ{{\EuScript J}}
\nc\bfk{{\boldsymbol k}}\nc\bfK{{\boldsymbol K}}\nc\cK{{\EuScript K}}
\nc\bfl{{\boldsymbol l}}\nc\bfL{{\boldsymbol L}}\nc\cL{{\EuScript L}}
\nc\bfm{{\boldsymbol m}}\nc\bfM{{\boldsymbol M}}\nc\cM{{\EuScript M}}
\nc\bfn{{\boldsymbol n}}\nc\bfN{{\boldsymbol N}}\nc\cN{{\EuScript N}}
\nc\bfo{{\boldsymbol o}}\nc\bfO{{\boldsymbol O}}\nc\cO{{\EuScript O}}
\nc\bfp{{\boldsymbol p}}\nc\bfP{{\boldsymbol P}}\nc\cP{{\EuScript P}}
\nc\bfq{{\boldsymbol q}}\nc\bfQ{{\boldsymbol Q}}\nc\cQ{{\EuScript Q}}
\nc\bfr{{\boldsymbol r}}\nc\bfR{{\boldsymbol R}}\nc\cR{{\EuScript R}}
\nc\bfs{{\boldsymbol s}}\nc\bfS{{\boldsymbol S}}\nc\cS{{\EuScript S}}
\nc\bft{{\boldsymbol t}}\nc\bfT{{\boldsymbol T}}\nc\cT{{\EuScript T}}
\nc\bfu{{\boldsymbol u}}\nc\bfU{{\boldsymbol U}}\nc\cU{{\EuScript U}}
\nc\bfv{{\boldsymbol v}}\nc\bfV{{\boldsymbol V}}\nc\cV{{\mathscr V}}
\nc\bfw{{\boldsymbol w}}\nc\bfW{{\boldsymbol W}}\nc\cW{{\mathscr W}}
\nc\bfx{{\boldsymbol x}}\nc\bfX{{\boldsymbol X}}\nc\cX{{\EuScript X}}
\nc\bfy{{\boldsymbol y}}\nc\bfY{{\boldsymbol Y}}\nc\cY{{\mathscr Y}}
\nc\bfz{{\boldsymbol z}}\nc\bfZ{{\boldsymbol Z}}\nc\cZ{{\EuScript Z}}
\nc{\1}{{\mathbbm{1}}}
\nc\rr{{\mathbb R}}
\nc\ee{{\mathbb E}}
\nc\sS{{\mathcal S}}
\nc{\integers}{{\mathbb Z}}
\nc{\ff}{{\mathbb F}}
\nc{\ii}{{\mathbb I}}
\nc{\sC}{{\mathfrak C}}
\nc{\sL}{{\mathfrak L}}
\nc\hH{{\mathsf H}}
\nc\gG{{\mathsf G}}

\nc{\remove}[1]{}

\DeclareSymbolFont{bbold}{U}{bbold}{m}{n}
\DeclareSymbolFontAlphabet{\mathbbold}{bbold}

\DeclareMathOperator{\rank}{rank}

\newcommand{\genstirlingII}[3]{%
  \genfrac{\{}{\}}{0pt}{#1}{#2}{#3}%
}

\newcommand{\stirlingII}[2]{\genstirlingII{}{#1}{#2}}

\nc\Renyi{R{\'e}nyi }
\allowdisplaybreaks

\usepackage[normalem]{ulem}
\newcommand\redout{\bgroup\markoverwith{\textcolor{red}{\rule[0.5ex]{2pt}{0.8pt}}}\ULon}
\newcommand\blueout{\bgroup\markoverwith{\textcolor{blue}{\rule[0.5ex]{2pt}{0.8pt}}}\ULon}
\newcommand{\since}[1]{{\text{({#1})}}\hspace{2em}}

\begin{document}
\title{R{\'e}nyi divergence guarantees for hashing with linear codes\thanks{
Madhura Pathegama is with Dept. of ECE and ISR, University of Maryland, College Park, MD 20742. Email: madhura@umd.edu. His research was supported in part 	by NSF grants CCF2104489 and CCF2330909.

Alexander Barg is with Dept. of ECE and ISR, University of Maryland, College Park, MD 20742. Email: abarg@umd.edu. His research was supported in part by NSF grants CCF2110113 (NSF-BSF), CCF2104489, and CCF2330909.
}}

	\author{Madhura Pathegama,~\IEEEmembership{Graduate Student Member,~IEEE,}
		and Alexander~Barg,~\IEEEmembership{Fellow,~IEEE}}
        
	\date{}

\maketitle

\begin{abstract}
We consider the problem of distilling uniform random bits from an unknown source with a given $p$-entropy using linear hashing. As our main result, we estimate the expected $p$-divergence from the uniform distribution over the ensemble of random linear codes for all integer $p\ge 2$. 
The proof relies on analyzing how additive noise, determined by a random element of the code from the ensemble, acts on the source distribution. This action leads to the transformation of the source distribution into an approximately uniform one, a process commonly referred to as distribution smoothing.
We also show that hashing with Reed-Muller matrices reaches intrinsic randomness of memoryless Bernoulli sources in the $l_p$ sense for all integer $p\ge 2$.
\end{abstract}


 \thispagestyle{empty}

\section{Introduction}

Uniform random bit-strings are an essential resource in both computer science and cryptography. 
In computer science, numerous algorithms rely on randomization to efficiently solve problems.  These algorithms encompass a wide range of applications, including primality testing \cite{miller1975riemann,rabin1980probabilistic}, solving graph-based problems \cite{karger1993global,luby1985simple}, and random sampling \cite{metropolis1953equation,vitter1985random}.
Moreover, uniform random bits are indispensable in many cryptographic applications. 
Examples include universal hash functions \cite{carter1977universal}, pseudo-random generators \cite{haastad1999pseudorandom},  randomized encryption schemes \cite{rivest1983randomized}, and randomized signature algorithms \cite{waters2009dual}. 
In essence, uniform random bits provide the necessary unpredictability and confidentiality required for robust cryptographic systems. 
They safeguard sensitive information and thwart attacks, making them an essential resource.

One way of distilling a uniform distribution is to use a weak random source, i.e., a source with low entropy, and convert its randomness to a uniform distribution. 
For a weak random source with a known distribution, the amount of uniform bits that can be distilled from the source is called the intrinsic randomness of the source \cite{vembu1995generating}, and it is possible to construct a deterministic function that transforms almost all the entropy of the source to uniform $q$-ary symbols. 
In many cryptographic problems, however, the distribution of the source remains unknown. Usually, we assume that we know some quantitative measure of randomness, for instance, the entropy of the source.
A (random) function that is capable of distilling uniform bits from an unknown source $Z$ with a given (R{\'e}nyi) entropy
$H_p(Z)$ is called a randomness extractor \cite{nisan1996extracting}. 
If the inherent randomness of this function is nearly independent of the output, such a function is called a strong extractor in the literature. 
Note that, almost always, it is impossible to convert all the randomness of a weak source into a perfect uniform distribution; what we usually obtain is an approximation.  
Proximity to uniformity can be measured in several ways. For instance, in many 
cryptographic applications, the distance between probability distributions is often
measured using the total variation distance or KL divergence.

A classic result in randomness extraction, known as the leftover hash lemma or LHL \cite{impagliazzo1989pseudo,tyagi2023information}, 
implies that universal hash functions are capable of converting almost all min-entropy of the source to a bit string whose distribution is close to uniform in the sense of total variation distance. The assumption of this lemma can be relaxed from the min-entropy to ${p}$-R{\'e}nyi entropy in \cite{cachin1997entropy} and (for $p=2$) in \cite{bennett1995generalized}. 
In some instances, uniformity guarantees of the total variation distance or KL divergence are 
insufficient for the applications 
\cite{canetti2006mitigating,kaslasi2021public,dodis2004possibility,skorski2015shannon}. 
Motivated by these shortcomings, the authors of \cite{hayashi2016equivocations} obtained 
asymptotic $p$-\Renyi divergence-based ($p \in [0,2]$) uniformity guarantees for universal hash 
functions. Extending their work, these authors further published \cite{tan2018analysis}, which analyzes linear hash functions using higher-order \Renyi entropies. The focus of \cite{tan2018analysis} is on
the remaining uncertainty of the source conditioned on the hash value and the seed, and it does not
address uniformity guarantees.

In this paper, we focus on distilling an approximately uniform distribution from a weak source with the help of a special class of hash functions, namely random linear codes $\{\cC\}$, 
relying on the $p$-norms, $p=2,3,\dots,\infty$ to quantify the distance to uniformity. 
Our work extends the results of \cite{hayashi2016equivocations} to higher-order $p$-norms. Our stronger
hashing guarantees are obtained based on strengthened assumptions concerning the randomness of the source.

An early study of linear codes as hash functions was performed by Alon \emph{et al.} \cite{alon1999linear}. In particular, within certain parameter regimes, \cite{alon1999linear} derived bounds for the expected {\em maximum bucket size}, 
i.e., the largest set of source sequences having the same syndrome. 
This result effectively compares the distribution of hash values to 
the uniform distribution in $l_\infty$ sense. Extending this work, Dhar and Dvir 
\cite{dhar2022linear} showed that for almost all the linear hash functions, the size of the largest hash bucket is close to the 
expected size. 

As our main result, we estimate the expected 
$l_p$-distance (equivalently, the $p$-\Renyi divergence) of the hashed source to the uniform distribution on $\ff_q^m$. 
The assumption of knowing $H_p(Z)$ is weaker than the conditions in \cite{alon1999linear} and \cite{dhar2022linear}. In terms of the 
bounds on the distance to uniformity, previous results covered only the cases $p=1,\infty;$ here we address all the intermediate 
integer values. Extending our bounds for finite values of ${p}$ to $p=\infty,$ we also estimate the size of the largest hash bucket 
that complements the existing results. 

The main tool we use to prove our results is often called smoothing of distributions. 
We say a distribution $P$ on $\ff_q^n$ is {\em $(\epsilon,{p})$-smoothable} with respect to $\cC$ if $\|q^n P \ast P_{\cC}\|_{p}-1 \leq \epsilon,$ where $P_{\cC}$ is the uniform distribution over the code $\cC$. This metric is easily seen to measure 
closeness to uniformity; as explained formally in Sec.~\ref{sec:proxi}, it is also equivalent to $p$-\Renyi divergence.
Smoothing of distributions \cite{debris2023smoothing,micciancio2007worst} has also been studied in information theory under the names of channel resolvability \cite{han1993approximation},  
or noisy functions \cite{samorodnitsky2016entropy,samorodnitsky2019upper}. 
It has applications in  
information-theoretic security \cite{hayashi2006general,bloch2013strong,pathegama2023smoothing}, coding theory \cite{hkazla2021codes,rao2024criterion,pathegama2023smoothing}, converse coding theorems of information theory \cite{arimoto1973converse,polyanskiy2010arimoto}, strong coordination \cite{cover2007capacity, bloch2013strong}, secret key generation \cite{chou2015polar, luzzi2023optimal}, and worst-to-average case reductions in cryptography \cite{micciancio2007worst,brakerski2019worst,debris2022worst}. 

In this work, we use smoothing as a way to distill a uniform distribution from a weak (low-entropy) source. We start
with an ensemble of random linear codes $\{\cC\}$. Consider two random vectors $X$ and $Z$, where $X$ is distributed uniformly over
a code $\cC$ sampled from the ensemble and $Z$ is distributed according to $P$. If $P$ is smoothable with respect to $\cC$, then $X+Z$ is approximately uniformly distributed. Pursuing strong extraction, we also require that 
the (nearly) uniform vector $X+Z$ be independent of the code $\cC$. We show that by projecting the output random vector onto the syndrome space, we can eliminate the randomness borrowed from the choice of the codeword. We also show that the resulting random vector is almost independent of the choice of the code itself.

In addition to general sources, we consider the binary case, whereby the random source generates a sequence of independent (biased) Bernoulli trials. Since this is a specific source model, we can extract uniform
bits relying on a deterministic extractor. This approach was taken in earlier works  \cite{vembu1995generating,han1993approximation}. The authors of \cite{yu2019simulation} considered this problem under R{\'e}nyi divergence proximity measures and obtained the largest rate of uniform bits that can be produced in such instances, called the intrinsic randomness. Their results rely on a computationally involved rearrangement of the distributions. Seeking a low-complexity deterministic bit distillation mechanism,
we rely on bit extraction with parity-check matrices of binary Reed-Muller codes. This code family
is known to share many properties of random code ensembles \cite{abbe2023reed}.
In our work we show that almost all randomness (\Renyi entropy) of the Bernoulli source
can be converted into a uniform distribution using a low-complexity deterministic mechanism, given by the parity-check matrices of Reed-Muller (RM) codes. In other words, we show that RM codes achieve $p$-R{\'e}nyi resolvability rates of binary Bernoulli sources for integer $p\ge 2$.

\section{Preliminaries}

We begin with setting up the notation for the rest of the paper. Let $q\ge 2$ be a prime power and let $\ff_q^n$ be the $n$-dimensional vector space over the field $\ff_q$, equipped with the Hamming metric. A linear subspace of 
$\ff_q^n$ is called a linear code. Below we use the letter $\cC$ to refer to various codes considered in the paper, including
random codes. Unless specified otherwise, the codes are always assumed to be linear.

Probability distributions in this paper are supported on linear spaces over the
finite field $\ff_q$. For a random variable/vector $Z$, $P_Z$ denotes the probability mass function (pmf) of $Z$. Sometimes we use $Z$ to refer to random vectors from $\ff_q^n,$ using the same notation $P_Z$ for
their distributions. For a given probability distribution $P$, we write $Z\sim P$ to express the fact that $P_Z=P$.  If $Z$ is distributed uniformly over a subset $S\subset \ff_q^n$, with some abuse of notation, we write $Z \sim S$. 
The uniform random variable on $\ff_q^n$ is denoted by $U_n$, and $P_{U_n}$ refers to its distribution.
For a code $\cC \subset \ff_q^n$,  $P_{\cC}$ denotes the uniform distribution on the code $\cC$ and $X_{\cC}$ is a
random codeword of $\cC$.

\subsection{Measures of randomness} Let $X$ be a random variable. Our measure of choice for quantifying the amount of randomness in $X$ is the R{\'e}nyi entropy $H_p$ of order $p.$ For $p\in(1,\infty)$, it is defined as follows
 \begin{align*}
   H_{{p}} (X)&=\frac1{1-{p} }\log_q\Big(\sum_x P_X(x)^{{p}} \Big). 
 \end{align*}  
Taking the limits, we also find 
   \begin{align*}
		H_1(X)&=-\sum_x P_X(x)\log_q {P_X(x)}\\
           H_{\infty}(X)&= \min_x (-\log_q {P_X(x)}),  
\end{align*}
so for $p=1,$ the R{\'e}nyi entropy is the familiar Shannon entropy. If this is the case, below we omit the subscript. The entropy $H_\infty$ is often called min-entropy in the literature. 

  We note that $H_{p}(X)$ is a decreasing function of ${p}$, while the function $\frac{{p}-1}{{p}}H_{p}(X)$ increases with ${p}$. These relations allow us 
to bound R{\'e}nyi entropies of different orders through each other, leading to general bounds for uniformity
of smoothed sources.

Extending the notion of R{\'e}nyi entropy to a pair of distributions, one also defines the {\em R{\'e}nyi divergence}. 
For two discrete distributions $P\ll Q$  on the same probability space $\cX$ and for $p\in(1,\infty)$, let
    $$
D_{{p}}(P\|Q)=\frac{1}{{p} -1} \log_q \sum_x P(x)^{{p}} Q(x)^{-({p} -1)}.
   $$
As above, we can take the limits to obtain
   \begin{align*}
		D_1(P\|Q)&=\sum_x P(x)\log_q \frac{P(x)}{Q(x)}\\[.1in]
          D_\infty(P\|Q)&=  \max_x \log_q \frac{P(x)}{Q(x)} .
  \end{align*}
If $p=1$, below we omit the subscript and simply write $D(\cdot\|\cdot)$ to refer to the KL divergence.
\subsection{Proximity of distributions and measures of uniformity} \label{sec:proxi}
Since we are interested in quantifying the closeness of distributions to $P_{U_n}$, we start with introducing some measures of proximity. Before defining them, recall the expression for the $p$-norm ($p$-th moment) of a function $f:\ff_q^n\to\rr$:
   \begin{align*}
    \|f\|_{{p}}  = 
    \begin{cases}
        (\frac{1}{q^n}\sum_{x \in \ff_q^n}{|f(x)|}^{{p}} )^{1/{{p}} } & \text{ for } {p} \in {(0,\infty)}\\[.1in]
    \max_{x\in \ff_q^n}{|f(x)|} & \text{ for } {p} =\infty.
    \end{cases}
\end{align*}
Under the uniform distribution $U_n$ on $\ff_q^n$, $f$ becomes a random variable, and
 $\|f\|_p=(\ee_{P_{U_n}}|f|^p)^{\frac1p},$ $0<p<\infty$. If $f=P$ is a pmf on $\ff_q^n,$ then
   \begin{align}\label{eq:norm-H}
   \|P\|_p^p&=q^{-(p-1)H_p(P)-n}\\
   \|P\|_1&=q^{-n}. \label{eq:n1}
   \end{align}
The case of $p=1$ leads to a distance between distributions used in many applications. Define the
{\em total variation distance} as 
\begin{align*}
    d_{\text{TV}}(P,Q) = \max_{A\in \ff_q^n}|P(A)-Q(A)|.
\end{align*}
Its relation to the $1$-norm becomes apparent once we write it in a different form:
  $$
  d_{\text{TV}}(P,Q)=\frac{q^n}{2}\|P-Q\|_1.
  $$

There are more than a few options to measure the distance from a given distribution $P$ to the uniform distribution. We will use the following three metrics. 
The first one, which we also call $D_p$-{\em smoothness},
is the $p$-divergence, $D_p(P\|P_{U_n}), p\in[0,\infty].$
It is easy to see that $D_{{p}}(P\|P_{U_n}) = n-H_{{p}}(P)$, so it is an increasing function of ${p}$. 

Another measure of uniformity is the $l_p$ distance $\|P-P_{U_n}\|_{p}$. For computational convenience, we 
remove the dependence on $n$, normalizing by the expectation $\ee P=\|P\|_1,$ so we will
use $q^n\|P-P_{U_n}\|_{p} = \|q^nP-1\|_{p}$.

To introduce our last version of closeness to uniformity, observe that
\begin{align*}
    \|q^nP\|_{p} &=  \frac{\|P\|_{{p}} }{\|P\|_1 }\geq 1 \quad \text{ for }  {p} \in (1,\infty]
\end{align*}
with equality iff $P=P_{U_n}$. Thus, the better the distribution $P$ approximates the uniform distribution, the closer is $\|q^nP\|_{{p}} $ to 1.  Therefore, $\Delta^{(n)}_{p}(P):=\|q^nP\|_{{p}}-1$, $p\ne 1$  can be considered as another measure of uniformity. We call it the $l_{{p}}$-{\em smoothness} of the distribution $P$.

We say that the uniformity measures $m_1(P)$ and $m_2(P)$ are {\em equivalent} if $m_1(P)\le \epsilon$
implies that $m_2(P)=O_{p}(\epsilon)$ and the same holds upon interchanging the roles of 
$m_1$ and $m_2$. In the next proposition, we show that the measures introduced above are equivalent for ${p} \in (1,\infty)$. In one part of the proof, we rely on Clarkson's inequalities \cite[p.~388]{Simon2015}, which state that for functions $f$ and $g$ on $\ff_q^n,$
   \begin{align}
   \|{\textstyle{\frac12}}(f+g)\|_p^{p'}+\|{\textstyle{\frac12}}(f-g)\|_p^{p'} &\le
        \Big({\textstyle{\frac12}}\|f\|_p^p+{\textstyle{\frac12}}\|g\|_p^p\Big)^{{p'}/p}, \label{eq:C1}
         \\ \hspace*{1in} 1<p<2, \nonumber\\
       \|{\textstyle{\frac12}}(f+g)\|_p^p+\|{\textstyle{\frac12}}(f-g)\|_p^p &\le {\textstyle{\frac12}} \label{eq:C2}
       (\|f\|_p^p+\|g\|_p^p),\\ 
       \hspace*{1in} 2\le p<\infty, \nonumber
           \end{align}
where ${p'}=p/(p-1)$.

\begin{proposition}\label{prop:equiv} Let $P$ be a distribution on $\ff_q^n$ and let ${p} \in (1,\infty)$. Then 
      \begin{subequations}
        \begin{align}    
   \Delta^{(n)}_{p}(P) \le \epsilon &\;\Rightarrow\; D_p(P\|P_{U_n}) \leq 
   \frac{{p}}{{p}-1}\log_q(1+\epsilon)\label{eq:m1m2a}\\
    D_p(P\|P_{U_n})\le \epsilon
   &\;\Rightarrow\;\Delta^{(n)}_{p}(P) \le q^{\frac{{p}-1}{{p}}\epsilon}-1\label{eq:m1m2b}
   \end{align}
   \end{subequations}
and
      \begin{subequations}
        \begin{align} 
      \Delta^{(n)}_{p}(P) \le \epsilon   &\;\Rightarrow\;
      \|q^nP-1\|_p \leq \phi_{p}(\epsilon)\label{eq:m2m3a}\\
      \|q^nP-1\|_p\le \epsilon &\;\Rightarrow\;
      \Delta^{(n)}_{p}(P) \le \epsilon,\label{eq:m2m3b}
     \end{align}
   \end{subequations}
     where 
    $$
    \phi_{p}(\epsilon)=\begin{cases}
        2\Big(\Big(\frac{(1+\epsilon)^{{p}}+1}{2}\Big)^{{p}^\prime/{p}}-1\Big)^{1/{p}^\prime} &1<{p}<2\\[.1in]
        2\big(\frac{(1+\epsilon)^{{p}}-1}{2}\big)^{1/{p}}&2\le{p}<\infty.
    \end{cases}
   $$
\end{proposition}
\begin{proof}
To prove relations \eqref{eq:m1m2a}-\eqref{eq:m1m2b}, all we have to do is to write the divergence in a slightly different form:
\begin{align}\label{eq: Renyi smoothness}
D_{{p}}(P\|P_{U_n}) =
\begin{cases}
    \frac{{p}}{{p} -1}\log_q \|q^nP\|_{{p}} \\
       \hspace*{.5in}\text{ for } {p} \in (0,1) \cap (1,\infty)\\
    \lim_{{p}'\to 1}\frac{{p}'}{{p}' -1}\log_q \|q^nP\|_{{p}'} \\ \hspace*{.5in} \text{ for } {p} = 1 \\
    \log_q\|q^nP\|_{\infty} \\   \hspace*{.5in}\text{ for } {p} = \infty.
\end{cases}
\end{align}
Let us show relations \eqref{eq:m2m3a}-\eqref{eq:m2m3b}. Let ${p}\in(1,2), {p}'={p}/({p}-1)$, then from \eqref{eq:C1} we have
\begin{align}
    1 + \left\|\frac{q^nP -1}{2}\right\|_{{p}}^{{p}'}  &\leq 
                \left\|\frac{q^nP +1}{2}\right\|_{{p}}^{{p}'} + \left\|\frac{q^nP -1}{2}\right\|_{{p}}^{{p}'} \nonumber \\
    &\leq \Big(\frac{1}{2}(\|q^nP\|_{{p}}^{{p}} + 1)\Big)^{{p}'/{p}}, \nonumber 
\end{align}
or
\begin{align*}
    \|{q^nP -1}\|_{{p}} \leq  2\Big(\Big(\frac{\|q^nP\|_{{p}}^{{p}} + 1}{2}\Big)^{{p}'/{p}}-1\Big)^{1/{p}'}.
\end{align*}
For ${p} \in [2, \infty)$, using \eqref{eq:C2}, we obtain
\begin{align}
    1 + \left\|\frac{q^nP -1}{2}\right\|_{{p}}^{{p}}  &\leq \left\|\frac{q^nP +1}{2}\right\|_{{p}}^{{p}} + \left\|\frac{q^nP -1}{2}\right\|_{{p}}^{{p}}\nonumber\\
    &\leq \frac{1}{2}(\|q^nP\|_{{p}}^{{p}} + 1).\label{eq: Clark_2}
\end{align}
This yields
\begin{align*}
    \|{q^nP -1}\|_{{p}} \leq 2\Big(\frac{\|q^nP\|_{{p}}^{{p}} -1}{2}\Big)^{1/{p}}.
\end{align*}
Finally, for ${p} \in [1,\infty)$, by Minkowski's inequality we have 
   $$
    \|q^nP\|_{p} -1 \leq \|q^nP-1\|_{{p}}. \qedhere
   $$

\end{proof}

This equivalence allows us to choose the most convenient metric, and below we formulate our results based on the $l_p$-smoothness.

Observe that for ${p} > 1$, the $D_{{p}}$- or $l_{{p}}$-smoothness can be treated as a stronger measure of uniformity compared to the total variation distance because convergence in $l_p$ norm is stronger than convergence in $l_1$ norm. Therefore, our work can be interpreted as characterizing sufficient conditions for extracting uniformity in a strong sense compared to the TV distance.  

\subsection{Universal hash functions.} In this section, we recall the known results for 
extracting uniformity based on universal hash functions. 

Below $Z$ is a random vector on $\ff_q^n$ with some underlying probability
distribution $P_Z$.
\begin{definition}\label{def:extractor1}
A set $\cF_{n,m}$ of functions from $\ff_2^n \to \ff_2^m$ 
 is said to form a \emph{universal hash family} (UHF) if for all $u, v \in \ff_2^n$ with $u\neq v$, $\Pr_{f\sim \cF}(f(u)=f(v))\leq \frac{1}{2^m}$.
\end{definition}

The randomness extraction property of hash functions relies on the following classic result
 \cite{impagliazzo1989pseudo}, \cite[p.122]{tyagi2023information}.

\begin{theorem}\label{thm: LHL}
{{\rm (Leftover Hash Lemma)}} \label{prop: lhl}
Let $\epsilon >0$ and let $n,m,t$ be positive integers. 
Let $\cF_{n,m}$ be a UHF and $f\sim \cF_{n,m}$. If $m\le H_\infty(Z)-2\log(1/\epsilon),$ then,  
\begin{align}\label{eq: LHL}
    d_{\text{\rm TV}}(P_{f(Z),f}, P_{U_m} \times P_{f}) \leq \epsilon/2,
\end{align}
\end{theorem}

Note that the condition in \eqref{eq: LHL} can also be written as follows:
\begin{align*}
    E_{f\sim \cF}\|q^m P_{f(Z)}-1\|_1 \leq \epsilon.
\end{align*}

We note another well-known result \cite{bennett1995generalized} that established the randomness extracting property of hash functions in a different sense. Rephrased to match our notation, it has the following form.
\begin{theorem}[\cite{bennett1995generalized}]\label{thm: bennet}
With the notation of Theorem \ref{thm: LHL}, if $m\le H_2(Z)-\log(1/\epsilon)$, then  
\begin{align*}
    E_{f\sim \cF}[D(P_{f(Z)}\|P_{U_m})] \leq \frac{\epsilon}{\ln2}.
\end{align*}
\end{theorem}
Estimates similar to this theorem that apply to randomness extractors other than
2-universal hash functions were presented in \cite{fehr2008randomness}.

The uniformity measures in Theorem~\ref{thm: bennet} and in the leftover hash lemma are, respectively, the KL divergence
and the total variation distance, and they are essentially equivalent because of Pinsker's inequalities.
At the same time, Theorem \ref{thm: bennet} relies on a somewhat weaker measure of randomness, namely,
the $2$-R{\'e}nyi entropy, while yielding essentially the same uniformity claim, so it forms a slightly stronger
claim than Theorem \ref{prop: lhl}.

Pursuing the line of thought expounded in the introduction, we aim to characterize uniformity
in a more stringent way by moving from the TV distance to ${p}$-norms with ${p}>1$.

\section{Smoothing-based randomness extraction}\label{sec: Sm}
In this section, we show that random linear codes over $\ff_q$ are capable of extracting randomness measured by R{\'e}nyi entropy. Below the parameters of the code are written as $[n,k]_q$, where $n$ is the length and $k$ is the dimension of the code. For positive integers $m\le n$, denote by $\sC$ the set of all $[n,n-m]_q$ linear codes, and for a code $\cC\in \sC$ denote by $\hH$ its parity-check matrix (for brevity, we suppress
the dependence on the code from the notation). We assume that $\hH$ is of dimensions $m\times n$ and 
note that $\rank(\hH)=m.$ If $\cC$ is a random code and $Z$ is a random vector, then $P_{\hH Z}$ denotes the 
induced distribution on $\ff_q^m$.

The following theorem forms the main result of our work.  
\begin{theorem}\label{thm: main}
Let $\epsilon >0$ and let ${p}\ge 2$ be an integer. 
If $Z$ is a random vector from $\ff_q^n$ with $Z\sim P_Z$ and $ m \leq H_{p}(Z) - {p}-\log_q(1/\epsilon)$, then
\begin{align*}
  \ee_{\cC \sim \sC}[\Delta^{(m)}_{p}(P_{\hH Z})]\le \epsilon.
\end{align*}
\end{theorem}

The following corollary provides a similar result for the other two types of uniformity measures 
defined above.

\begin{corollary}\label{cor: main}
With the assumptions of Theorem~\ref{thm: main}, we have
    \begin{align}
 \ee_{\cC \sim \sC}[ D_p(P_{\hH Z}\|P_{U_m})] &\leq \frac{{p}\epsilon}{({p}-1)\ln q} \label{eq:D1}\\
    \ee_{\cC \sim \sC}\|q^mP_{\hH Z}-1\|_p &\leq 2^{1-1/{p}}((1+\epsilon)^{p} -1)^{1/{p}}. \label{eq:D2}
\end{align}
\end{corollary}

\begin{proof}
    Using the inequality $\ln x\le x-1$ in the first relation in \eqref{eq: Renyi smoothness}, we find that 
    $$
    \frac{{p}-1}{{p}}(\ln q) D_{p}(P_{\hH Z}\|P_{U_m}) \leq \|q^mP_{\hH Z}\|_{p}-1,
    $$ 
{which} proves \eqref{eq:D1}. The second inequality is proved as follows. From \eqref{eq: Clark_2},
    \begin{align*}
        \Big[2\Big\|\frac{q^mP_{\hH Z}-1}{2}\Big\|_p^p+1\Big]^{1/p} \leq \|q^mP_{\hH Z}\|_p
    \end{align*}
From the assumptions, we have
\begin{align*}
    \ee_{\cC \sim \sC}\Big[2\Big\|\frac{q^mP_{\hH Z}-1}{2}\Big\|_p^p+1\Big]^{1/p} \leq 1+\epsilon.
\end{align*}
For brevity, let $X:=2^{1/p}\Big\|\frac{q^mP_{\hH Z}-1}{2}\Big\|_p$, then the above inequality takes the form
   $$
   \ee_{\cC \sim \sC}[X^p+1]^{1/p} \leq 1+\epsilon.
   $$ Applying Jensen's inequality we obtain 
\begin{align*}
    ([\ee_{\cC \sim \sC}X]^p+1)^{1/p}\leq \ee_{\cC \sim \sC}[X^p+1]^{1/p} \leq 1+\epsilon.
\end{align*}
This implies the relation,
\begin{align*}
    \ee_{\cC \sim \sC}X \leq ((1+\epsilon)^p-1)^{1/p},
\end{align*}
which proves \eqref{eq:D2}.
\end{proof}

\begin{remark}\label{rm: ind}
The condition $\ee_{\cC \sim \sC}[\|q^mP_{\hH Z}-1\|_p] 
 \leq \epsilon$ not only guarantees that the distribution $ P_{\hH Z}$ is close to uniform on average, but also says that $P_{\hH Z}$ is almost independent of the choice $\cC \sim \sC$. Indeed, since $P_{\hH Z}$ is close to the uniform distribution, it
 does not depend on $\cC$ for almost all the codes.
 As mentioned in the introduction, the uniformity of hash functions is typically measured in total variation distance. As a result, the approximate independence is also measured in the same metric. Since $\|q^mP_{\hH Z}-1\|_p$ is increasing in $p$, it is a stronger measure of independence compared to the total variation distance. \hfill$\lhd$
\end{remark}
The main ingredient of the proof of Theorem \ref{thm: main} is the following result on smoothing a source using random linear codes.
\begin{theorem}\label{thm: smooth}
    Let $Z$ be a random vector in $\ff_q^n$. Let $\sC$ be the set of all $[n,k]_q$ linear codes. Then for 
    all natural $p\ge 2$,
     \begin{align}\label{eq: smooth}
        \ee_{\cC \sim \sC}[\|q^n P_{X_\cC+Z}\|_{{p}}^{{p}}] \leq 
        \sum_{d=0}^{{p}}\binom{{p}}{d}q^{({p}-d)(d+n-k - H_{{p}}(Z))},
    \end{align}
    where $X_\cC$ is a uniform random codeword of $\cC$.
\end{theorem}

\begin{remark}\label{remark:linear} This result is independent of its applications to hashing, and it can be strengthened if random linear
codes are replaced with the ensemble of all binary codes of the same cardinality. More precisely, if $\sC^\prime$ is the set of 
all $(n,q^k)_q$ codes, then
    \begin{align}\label{eq:nonlinear}
        \ee_{\cC \sim \sC^\prime}[\|q^n P_{X_\cC+Z}\|_{{p}}^{{p}}] \leq 
        \sum_{d=0}^{{p}}\stirlingII{{p}}{d}q^{({p}-d)(n-k - H_{{p}}(Z))}, 
    \end{align}
    where $\stirlingII{{p}}{d}$ is the Stirling number of second kind. From \cite[Theorem 3]{rennie1969stirling}, we have $\stirlingII{{p}}{d} \leq \binom{{p}}{d}d^{{p}-d} \leq \binom{{p}}{d}q^{d({p}-d)}$, which shows that \eqref{eq:nonlinear} is a better bound than
    \eqref{eq: smooth}. 
Since this result does not have a direct implication for hashing, we will omit the proof. \hfill$\lhd$
\end{remark}

Below we will say that an unordered tuple of vectors $(v_1,\dots,v_{p})$ is a {\em $[{p},d]$-tuple} if
their linear span has dimension $d.$ If $[v_1, \dots, v_{p}]$ is the $n \times {p}$ matrix 
whose columns are the vectors $v_1, \dots, v_{p}$, then this means that $\rank[v_1,\dots,v_{p}] = d$. 
To shorten the notation, below we will write ``$(v_1^p\text{ is }[p,d])$''.

The proof of Theorem \ref{thm: smooth} is based on the next two lemmas.

\begin{lemma}\label{lemma: balanced}
    Let $\sC$ be the set of $[n,k]_q$ linear codes in $\mathbbm{F}_q^n$. Let $p\ge 2$ be an integer. For any function $f : (\ff_q^n)^p \to [0,\infty)$,
    \begin{multline}\label{eq: balanced}
        \ee_{\cC \sim \sC}\Big[\sum_{z_1,\dots,z_p\in\cC}f(z_1,\dots, z_{{p}})\Big]
        \\ \leq  \sum_{d=0}^{\min\{k,{p}\}}q^{d(k-n)}\sum_{(v_1^p\text{ is }[p,d])}f(v_1,\dots, v_{{p}}).
    \end{multline}   
\end{lemma}
\begin{remark}
    Lemma \ref{lemma: balanced} is equivalent to the following fact. For a fixed $u_1, \dots, u_{p} \in \ff_q^n$,
   \begin{align}\label{eq: list_dec}
        \Pr_{\cC \sim \sC}\{\{u_1,\dots,u_{p}\} \subset \cC\} 
        \leq& q^{-\rank[u_1,\dots,u_{p}](n-k)}.
    \end{align}

The ensemble of linear codes defined by random generator matrices with independent elements 
satisfies the above inequality with equality \cite{mosheiff2020ldpc}.
Our proof applies not only to this ensemble but also to full-rank random linear codes as in the lemma, and a broader class of  $p$-{\em balanced ensembles} (see Definition \ref{def:balanced}). As a result, $p$-divergence-based hash guarantees hold for any $p$-balanced ensemble of suitable length and dimension. 
Furthermore, some random ensembles, such as regular LDPC codes \cite{mosheiff2020ldpc} and randomly punctured low-bias codes \cite{guruswami2022punctured}, satisfy an approximate version of \eqref{eq: list_dec}. 
This suggests that these ensembles also provide hashing guarantees, albeit weaker than those established here.
\hfill$\lhd$
\end{remark}

The next lemma is the main technical contribution of our work. We will use it to bound the sums that arise in the expansion of the term  $ \ee_{\bfC}[\|q^n \bfP_{\bfC} \ast P_{\bfZ}\|_{\alpha}^{\alpha}]$.
\begin{lemma}\label{lemma: norm_bound}
    Let $1\leq d < {p}$. Then for any function $f : \mathbbm{F}_q^{n} \to [0,\infty)$,
    \begin{multline*}
        \sum_{(v_1^p\text{ is }[p,d])}\prod_{1\le j\le p}f(v_j)
        \\\leq \binom{{p}}{d}\sum_{m= 0}^{{p} -d} \binom{{p}-d}{m}(q^d-1)^{{p}-d-m}q^{nd}\\ \times f(0)^m\|f\|_1^{d-1}\|f\|_{{p}-d-m+1}^{{p}-d-m+1}.
    \end{multline*}
\end{lemma}

Proofs of Lemmas \ref{lemma: balanced} and \ref{lemma: norm_bound} are deferred to Appendix.

Next, assuming the validity of Lemmas \ref{lemma: balanced} and \ref{lemma: norm_bound}, let us give a proof of Theorem \ref{thm: smooth}.

\begin{proof} Recall that $P_Z$ refers to the distribution of the vector $Z$ sampled from $\ff_q^n$.
Below we abbreviate $\ee_{\cC \sim \sC}$ to $\ee_{\cC}$. We have
\begin{align}
    \ee_{\cC}&[\|q^n P_{X_\cC+Z}\|_{{p}}^{{p}}]=\mathbbm{E}_{\cC}[\|q^nP_{X_\cC+Z}\|_{p}^{p}] \notag\\
    &
     = \mathbbm{E}_{\cC}\Big[\frac{1}{q^n}\sum_{x \in\mathbbm{F}_q^n} (q^n P_Z \ast P_{\cC}(x))^{{p}}\Big]
     \notag\\
    & = \mathbbm{E}_{\cC}\Big[q^{n({p}-1)}\sum_{x \in\mathbbm{F}_q^n} 
    \sum_{z_1,\dots,z_p \in \ff_q^n} \prod_{j=1}^p P_Z(x-z_j)P_{\cC}(z_j)\Big]\notag\\
    & = \sum_{x \in \mathbbm{F}_q^n}\frac{q^{n({p}-1)}}{q^{k{p}}}\mathbbm{E}_{\cC}
    \biggl[\sum_{z_1,\dots,z_p\in\cC}  \prod_{l=1}^{{p}}P_Z(x-z_l)\biggr]\notag\\
    & \leq \sum_{d=0}^{p} q^{({p}-d)(n-k)-n}\sum_{x \in \mathbbm{F}_q^n}\sum_{(v_1^p\text{ is }[p,d])} \prod_{l=1}^{{p}} P_Z(x-v_l),  \label{eq:rank}
\end{align}
where \eqref{eq:rank} follows by Lemma \ref{lemma: balanced}. Let us introduce the notation
\begin{align*}
    g(d):=\frac1{q^n}\sum_{x \in \mathbbm{F}_q^n}\sum_{(v_1^p\text{ is }[p,d])} \prod_{l=1}^{{p}} P_Z(x-v_l).
\end{align*}
Then the bound \eqref{eq:rank} can be written as 
\begin{align*}
    \ee_{\cC}[\|q^n &P_{X_\cC+Z}\|_{{p}}^{{p}}] \leq \sum_{d=0}^{p} q^{({p}-d)(n-k)}g(d).
\end{align*}
The theorem will be proved if we show that
   \begin{equation}
    g(d) \leq \binom{{p}}{d}q^{({p}-d)(d-H_{p}(P_Z))}. \label{eq:g}
   \end{equation}
For $d=0,p$ this is immediate from the definition:
   \begin{align*}
       g(0)&= \frac{1}{q^n}\sum_{x\in \ff_q^n}P_Z(x)^{p} = q^{-({p}-1)H_{p}(P_Z)-n} \notag\\
    &\leq q^{-{p} H_{p}(P_Z)}\\
       g({p}) &\le\frac1{q^n}\sum_{x\in\ff_q^n}\prod_{l=1}^p\sum_{v\in\ff_q^n} P_Z(x-v) 
       = 1 = q^{({p}-{p})H_{p}(P_Z)}.
   \end{align*}

Now let us consider the case $0 < d < {p}$. To shorten the writing, for the remainder of the proof,
we put $\tau=p-d-m-1.$ Using Lemma \ref{lemma: norm_bound} and \eqref{eq:n1}, we have
\begin{align}
    g(d) 
    &= \frac{1}{q^n}\sum_{x \in \mathbbm{F}_q^n}\sum_{(v_1^p\text{ is }[p,d])} \prod_{l=1}^{{p}} P_Z(x-v_l) \nonumber\\
    &\leq \frac{1}{q^n}\sum_{x \in \mathbbm{F}_q^n} \binom{{p}}{d}\sum_{m= 0}^{{p} -d} \binom{{p}-d}{m}(q^d-1)^{{p}-d-m}q^{nd} \nonumber
    \\ &\hspace*{.8in} \times P_Z(x)^m\|P_Z\|_1^{d-1}\|P_Z\|_{\tau}^{\tau} \nonumber\\ 
    &= \binom{{p}}{d}\sum_{m= 0}^{{p} -d} \binom{{p}-d}{m}(q^d-1)^{{p}-d-m}\sum_{x}P_Z(x)^m\|P_Z\|_{\tau}^{\tau}.  \label{eq:0p}
\end{align}
First let $m=0$. Using \eqref{eq:norm-H},
   $$
   \sum_{x}P_Z(x)^m\|P_Z\|_{\tau}^{\tau} =q^n\|P_Z\|_{{p}-d-1}^{{p}-d-1}= q^{-({p}-d)H_{{p}}(P_Z)}.
   $$
Further for $m\ge 1$, again using \eqref{eq:norm-H} and the inequalities $m<p, \tau<p$, we obtain
    \begin{align*}
\sum_{x}P_Z(x)^m\|P_Z\|_{\tau}^{\tau}&= q^{-(m-1)H_{m}(P_Z)-(\tau+1)H_{\tau}(P_Z)-n}\\ 
&\leq q^{-({p}-d)H_{{p}}(P_Z)}.
    \end{align*}
Using these results in \eqref{eq:0p}, we obtain
\begin{align*}
    g(d) 
    &\leq \binom{{p}}{d}q^{-({p}-d)H_{{p}}(P_Z)}\sum_{m= 0}^{{p} -d} \binom{{p}-d}{m}(q^d-1)^{{p}-d-m},
\end{align*}
which gives the right-hand side of \eqref{eq:g}. This concludes the proof of Theorem \ref{thm: smooth}. \end{proof}

To prove Theorem \ref{thm: main}, we need an additional lemma that establishes a connection between smoothing and linear hashing.

\begin{lemma}\label{lem:sm_to_proj}
   Let $\cC$ be an $[n,k]_q$ linear code, and let $\hH$ be its parity check matrix. Let $X_\cC $ be a uniform random codeword of $\cC$ and $Z\sim P_Z$ be a random vector in $\ff_q^n$. Then for ${p} \in (0,\infty]$
   \begin{align*}
       \|q^{n-k} P_{\hH Z} \|_{p} = \|q^nP_{X_\cC+Z}\|_{p}.
   \end{align*} 
\end{lemma}

\begin{proof}
Define $\sS := \{\hH y : y \in \ff_q^n\}$ to be the set of syndromes corresponding to code $\cC$. 
    Observe that 
    \begin{align*}
        P_{\hH(X_\cC+Z)}(u) = \sum_{y: \hH y = u}P_{X_\cC+Z}(y)
        = |\cC|P_{X_\cC+Z}(y_u),
    \end{align*}
    where $y_u$ is a representative of the coset defined by $u$. The last equality is due to the fact that 
    $$ P_{X_\cC+Z}(y) = P_{\cC}\ast P_Z (y) = P_{\cC}\ast P_Z (y+c) = P_{X_\cC+Z}(y+c)$$ for any $c \in \cC$.
    Hence, for ${p} \in (0,\infty)$,
    \begin{align*}
        \|q^{n-k} P_{\hH(X_\cC+Z)} \|_{p}^{p} &= q^{(n-k)({p}-1)} \sum_{u \in \ff_q^{n-k}}P_{\hH(X_\cC+Z)}(u)^{p}\\
        &= q^{(n-k)({p}-1)} \sum_{y \in \ff_q^{n}}\frac{(|\cC|P_{X_\cC+Z}(y))^{p} }{|\cC|}\\
        &=\|q^nP_{X_\cC+Z}\|_{p}^{p}.
    \end{align*}
    Finally, since $\hH(X_\cC+Z)=\hH Z$, we have proven the desired statement for ${p} \in (0,\infty)$. 
    The case ${p} = \infty$ is obtained by taking a continuous limit.
\end{proof}

Now Theorem \ref{thm: main} follows from Theorem \ref{thm: smooth} and Lemma \ref{lem:sm_to_proj} as follows:
\begin{proof} (of Theorem \ref{thm: main}) By Lemma \ref{lem:sm_to_proj}
  \begin{align*}
    \ee_{\cC \sim \sC}[\|q^{m} P_{\hH Z} \|_{p}^{p}] 
    & = \ee_{\cC \sim \sC}[\|q^n P_{X_\cC+Z}\|_{{p}}^{{p}}] \\
     \since{\text{Theorem \ref{thm: smooth}}} &\leq \sum_{d=0}^{{p}}\binom{{p}}{d}q^{d({p}-d)}q^{({p}-d)(m- H_{{p}}(Z))} \quad \quad \\
    &\leq \big( 1+ q^{m-H_{p}(Z)+{p}}\big)^{p}.
    \end{align*}
Therefore, we have
\begin{align} \label{eq: dalpha}
    \ee_{\cC \sim \sC}[\|q^{m} P_{\hH Z} \|_{p}] -1
    & \leq q^{m-H_{p}(Z)+{p}}.
\end{align}
Choosing $m$ such that $m \leq H_{p}(Z)-{p}-\log_q(1/\epsilon)$ we obtain $\ee_{\cC \sim \sC}\|q^{m} P_{\hH Z}\|_{p} -1 \leq \epsilon$ proving the desired result.
\end{proof}

\begin{remark}\label{rm: loss} (Symbol loss) If it were possible to distill all the randomness of $H_p(Z)$ into uniform symbols,
this would yield $H_p(Z)$ random symbols. However, Theorem \ref{thm: main} (equivalently Corollary \ref{cor: main}) states that we can achieve the prescribed uniformity if we distill  $ m \leq H_p(Z)-p-\log_q(1/\epsilon)$  symbols. 
This is at least $p+\log_q(1/\epsilon)$ symbols away from the ideal target. 

More precisely, we say that hashing achieves {\em $(\epsilon,p)$-uniformity for a source $Z$ with a loss of $s=H_p(Z)-m$ symbols} if
\begin{align*}
  \ee_{\cC \sim \sC}[\Delta^{(m)}_{p}(P_{\hH Z})]\le \epsilon.
\end{align*}
In these terms, Theorem \ref{thm: main} states that $(\epsilon,p)$-uniformity is attained with an at most $p+\log_q(1/\epsilon)$ symbol loss.   
Since this is a constant, this loss is negligible in the setting of large $t$. 

However, when $H_p(Z)$ is not very large, this symbol loss may be significant.  
We show that it can be reduced for the case of collision entropy, i.e., $p=2,$ saving more than two symbols.

\begin{proposition} With the assumptions of Theorem \ref{thm: main}, linear hashing achieves $(\epsilon,2)$-uniformity 
with a $\log_q\frac{1}{(1+\epsilon)^2-1}$ loss.
\end{proposition}
\begin{proof}
Defining $(P_Z \circ P_Z)(x) = \sum_{y\in \ff_q^n}P_Z(y)P_Z(x+y)$, we proceed as follows:
\begin{align}
     &\ee_{\cC \sim \sC}[\|q^m P_{\hH Z} \|_{2}^2] 
    = \ee_{\cC \sim \sC}[\|q^n P_{X_\cC+Z}\|_{{2}}^{{2}}] \nonumber\\
     &= \frac{1}{q^n}\mathbbm{E}_{\cC\sim \sC} \sum_{y \in \ff_q^n}
     \sum_{x_1 \in \ff_q^n} q^nP_Z(y-x_1)P_{\cC}(x_1)\notag\\
    &\hspace*{.5in}\times\sum_{x_2 \in \ff_q^n} q^nP_Z(y-x_2)P_{\cC}(x_2) \nonumber\\
    &= \frac{1}{q^n}\mathbbm{E}_{\cC\sim \sC} \sum_{x_1 \in \cC}\sum_{x_2 \in \cC} \sum_{y \in \ff_q^n} \frac{q^nP_Z(y-x_1)}{|\cC|}\frac{q^nP_Z(y-x_2)}{|\cC|} \nonumber\\
    \intertext{(since $k=n-m$)}
     & = q^{n-2k}\mathbbm{E}_{\cC\sim \sC}\sum_{x_1 \in \cC}\sum_{x_2 \in \cC} (P_Z\circ P_Z)(x_1-x_2) \quad  \nonumber\\
     & = q^{n-k}\mathbbm{E}_{\cC\sim \sC}\sum_{z \in \cC} (P_Z\circ P_Z)(z) \nonumber\\
     & \leq q^{n-k}\big((P_Z\circ P_Z)(0) +  q^{k-n}\sum_{z \in \ff_q^n} (P_Z\circ P_Z)(z) \big) \nonumber\\
     & = q^{n-k}\big(\sum_{z \in \ff_q^n} P_Z(z)^2 +  q^{k-n} \big) \nonumber\\
     & = 1 + q^{m-H_2(Z)},\label{eq: impr_2}
\end{align}
where the inequality follows from Lemma~\ref{lemma: balanced}. Therefore if $m\leq H_2(Z)-\log_q\frac{1}{(1+\epsilon)^2-1}$, then 
   $$
    \ee_{\cC \sim \sC}[\Delta^{(m)}_{2}(P_{\hH Z})]\leq \epsilon. \qedhere
$$
\end{proof}

Similar improvements can be obtained for the first few integer values of $p\ge 3.$
At the same time, the proof of Prop.~\ref{prop: inft} suggests that for large $p$ it is desirable to have a loss which is a sublinear function of $p$. We do not know if our results can be tightened so as to attain this loss scaling. \hfill$\lhd$
\end{remark}

A natural question that arises in our line of work is whether our results on $D_p$- (or $l_p$-) smoothness improve the conclusions of classical theorems such as Theorems \ref{thm: LHL} or \ref{thm: bennet}.
The closeness to uniformity in these theorems is measured by the total variation distance and KL divergence (Shannon entropy), respectively. 
Note that these quantities are bounded by the $l_p$ distance and $D_p$-smoothness due to the monotonicity with respect to $p$. Specifically, we have 
\begin{align*}
    d_{\text{TV}}(P_{\hH(Z),\cC}\|P_{U_m}\times P_{\cC})&=
    2\ee_{\cC \sim \sC}\|q^m P_{\hH(Z)}-1\|_1 \notag\\
    &\leq 2\ee_{\cC \sim \sC}\|q^m P_{\hH(Z)}-1\|_p\\
    D(P_{\hH(Z)}\|P_{U_m} \times P_{\cC}) &= 
    \ee_{\cC\sim \sC} [D(P_{\hH(Z)}\|P_{U_m})] \notag\\
    &
    \leq \ee_{\cC\sim \sC} [D_p(P_{\hH(Z)}\|P_{U_m})].
\end{align*}
However, these inequalities do not enable us to improve upon the claim of Theorem \ref{thm: bennet}.
For instance, using $2$-entropy and $D_2$-smoothness in Theorem~\ref{thm: main}
 (or in \eqref{eq: impr_2}), and using the second
inequality essentially reproduces the claim of this theorem, which is already known for any UHF.
At the same time, this theorem does not imply our results for $p>2$ because the only assumption used in its proof is that of general UHFs. To establish our results, we additionally employ $p$-balancedness (defined in the appendix, see Def.~\ref{def:balanced}), which is not readily available under the general approach.

\subsection{Expected size of the largest hash bucket}\label{sec: largest bucket}
Given a source $Z$, the hashing function maps $n$-vectors generated by it to the syndromes of a randomly chosen $[n,n-m]$ linear code $\cC$. If $\cC$ is fixed, then the expectation with respect to $P_Z$ of the largest number of source sequences mapped to the same syndrome $s$ is $P_{\hH Z}(s)$. Therefore, the expected maximum bucket size is $\|P_{\hH Z} \|_\infty$. 
Taking a random code $\cC$ results in the expected maximum hash bucket having the 
size $\ee_{\cC \sim \sC}\|P_{\hH Z} \|_\infty$. The average size of the hash bucket is $q^{-m}$, and so we study the
quotient of these two quantities, given by $\ee_{\cC \sim \sC}[\|q^{m} P_{\hH Z} \|_\infty]$. 
Our main result is a new bound on this quantity when $m \leq H_\infty(Z)-n\epsilon$. Note that our main result, Theorem \ref{thm: main}, does not cover the case $p = \infty$. We will circumvent this obstacle by extending the bounds for finite $p$ to infinity.

Given a source $Z$, one would expect to distill $H_\infty(Z)$ approximately uniform $q$-ary 
symbols. Intuitively, if we aim at distilling fewer than $H_\infty(Z)$ random symbols, we can make the resulting vector to be `more' uniform. 
On the other hand, attempting to distill more than $H_\infty(Z)$ symbols forces the resulting distribution to deviate from uniformity. 
In the existing literature, this intuitive reasoning has been rigorously quantified for different regimes of $H_\infty(Z)$ and $m$ \cite{alon1999linear,dhar2022linear}. In terms of $\|q^mP_{\hH Z} \|_\infty$, the bound of \cite{alon1999linear} for 
the expected size of the largest hash bucket has the form
   \begin{align}\label{eq: Alon1}
        \ee_{\cC \sim \sC}[\|q^{m} P_{\hH Z} \|_\infty] = O(H_\infty(Z)\log H_\infty(Z)),
   \end{align}
where $m = H_\infty(Z)$.
\begin{remark} Both \cite{alon1999linear} and \cite{dhar2022linear} assume that the random vector $Z$ 
is produced by a `flat source', i.e., a uniform distribution over a subset. This assumption does
not entail a loss of generality because any source can be written as a convex combination of flat sources having the same min-entropy \cite[Lemma 6.10]{Vadhan2012Pseudo}, so 
their results, including \eqref{eq: Alon1}, \eqref{eq: Alon2} and \eqref{eq: DD}, hold for all sources with known min-entropy. At the same time, let us emphasize that this argument does not extend to other \Renyi entropies. \hfill$\lhd$
\end{remark}

For $m$ slightly smaller than the source entropy, namely $ m + \log m = H_\infty(Z)$, \cite{alon1999linear} has the bound
\begin{align}\label{eq: Alon2}
     \ee_{\cC \sim \sC}[\|2^{m} P_{\hH Z}\|_\infty] = O(\log H_\infty(Z)).
\end{align}
We consider the case of the gap between $m$ and the min-entropy proportional to $n$, showing that then the expected size is bounded by an absolute constant.
\begin{proposition}\label{prop: inft}
    Let $\epsilon >0$ and let $n$ be a positive integer. Let $Z\sim P_Z$ be a random vector in $\ff_q^n$. Choose $m$ such that $m \leq  H_\infty(Z)-n\epsilon$, and define $\sC$ the set of $[n,n-m]_q$ linear codes. Then
    \begin{align}\label{eq: L-inf_smooth}
    \ee_{\cC \sim \sC}[\|q^{m} P_{\hH Z} \|_\infty] = O_q(1).
\end{align}
\end{proposition}

\begin{proof}
From the fact $ H_{p}(Z) \leq \frac{{p}}{{p}-1}H_\infty(Z)$, we have
\begin{align}\label{eq: dinf}
    {p} D_\infty(P_Z\|P_{U_n}) -n \leq ({p}-1)D_{p}(P_Z\|P_{U_n}).
\end{align}
Further, from Lemma \ref{lem:sm_to_proj}
\begin{align*}
    \ee_{\cC \sim \sC}[\|q^{m} P_{\hH Z} \|_\infty] &= \ee_{\cC \sim \sC}[\|q^n P_{X_\cC+Z}\|_\infty]\\ 
    &= \ee_{\cC \sim \sC}[q^{D_\infty(P_{X_\cC+Z}\|P_{U_n})}] \\
    \since{\text{by \eqref{eq: dinf}}}&\leq  q^{\frac{n}{{p}}}\ee_{\cC \sim \sC}[q^{\frac{{p}-1}{{p}}D_{p}(P_{X_\cC+Z}\|P_{U_n})}] \quad \\
    &= q^{\frac{n}{{p}}} \ee_{\cC \sim \sC}[\|q^n P_{X_\cC+Z}\|_{{p}}]  \\
     \since{\text{by \eqref{eq: dalpha}}}&\leq q^{\frac{n}{{p}}}(1+q^{m-H_{p}(Z) + {p}}) \quad \quad\\
    &\leq q^{\frac{n}{{p}}}(1+q^{m-H_\infty(Z) + {p}}).
\end{align*}

Now, choose $m$ such that $m = H_\infty(Z)- n\epsilon$. Then by choosing ${p} = \epsilon n /2$, we obtain
\begin{align*}
    \ee_{\cC \sim \sC}[\|q^{m} P_{\hH Z} \|_\infty] &\leq q^{2/\epsilon}(1+q^{-\epsilon n/2}),
\end{align*}
 proving the desired result. 
\end{proof}

The recent work \cite{dhar2022linear} established a high probability estimate for the maximum size of a hash bucket. In particular, for the case $ m \leq H_\infty(Z)- n\epsilon$ 
their results imply that
\begin{align}\label{eq: DD}
    \Pr_{\cC \sim \sC}(\|q^{m} P_{\hH Z} -1 \|_\infty \geq \epsilon) \leq 2^{-\zeta n}
\end{align}
for some constant $\zeta = \zeta(\epsilon)$. We note that neither of the estimates \eqref{eq: L-inf_smooth} and \eqref{eq: DD} implies 
the other one, for the following reasons. The probabilistic estimate in \eqref{eq: DD} states that the fraction of codes that result in
an approximately uniform $P_{\hH Z}$ is close to one. This discounts the outliers that yield conditional distributions $P_{\hH Z}$
very different from the uniform one. On the other hand, the quantity  $\ee_{\cC \sim \sC}[\|q^{m} P_{\hH Z} \|_\infty]$ takes 
all these outliers into account, so \eqref{eq: DD} does not imply \eqref{eq: L-inf_smooth}. 
In the other direction, if we had a stronger estimate of the form $ \ee_{\cC \sim \sC}[\|q^{m} P_{\hH Z} -1\|_\infty] = o(1)$, Markov's 
inequality would yield a result comparable to \eqref{eq: DD}. As it stands, the estimate in \eqref{eq: L-inf_smooth} is not strong enough to imply \eqref{eq: DD}. 

An interesting open question is whether it is possible to improve \eqref{eq: L-inf_smooth} to $\ee_{\cC \sim \sC}[\|q^{m} P_{\hH Z} \|_\infty] = 1+o(1)$. If true, it would give strong guarantees for both uniformity and independence (see Remark \ref{rm: ind}).  

\begin{remark}
In the language of smoothing, $\|q^n P_{X_\cC+Z}\|_\infty -1$ is the $l_\infty$-smoothness of $P_{X_\cC+Z}$. In \cite{pathegama2023smoothing}, we proved that smoothing of $P_Z$ with respect to a random code $\cC$ of dimension $k \geq n-H_\infty(P_Z)+\epsilon n$ yields $\ee_{\cC}[\|q^n P_{X_\cC+Z} \|_\infty] = 1 + o_n(1)$. We do not see a way of obtaining a similar estimate for linear
codes, leaving this as an open problem. \hfill$\lhd$
\end{remark}

\section{Random bits from Bernoulli sources with RM codes}\label{sec: RM}

In addition to the problem of smoothing an arbitrary source, one can also consider the same question for a source
with a known distribution. 
For a fixed source, the amount of randomness that can be converted to (approximate) uniformity is sometimes called the intrinsic randomness \cite{vembu1995generating}. 
The authors of \cite{yu2019simulation} proposed 
a way of approximately simulating a given memoryless distribution with a known distribution, including the uniform one,
using another memoryless source and a carefully chosen mapping, and measuring proximity by the R{\'e}nyi divergence.
We will state their result for Bernoulli sources, starting with the definition of intrinsic randomness
of a Bernoulli random variable.
\begin{definition}\label{def: int_ran}
    Let $Z$ be a Bernoulli($\delta$) random variable. The $p$-R{\'e}nyi intrinsic randomness of $Z$ (in more detail, of a
memoryless source given by $Z$) is defined as
    \begin{multline*}
         I_p(Z) = \sup\Big\{{R}:\lim_{n\to\infty}\sup_{(k_n):\frac{k_n}n\to R}\\ \sup_{f_n\in F(n,k_n)}D_p(P_{f_n(Z_n)}\|P_{U_{k_n}}) = 0\Big\},
    \end{multline*}
    where $F(n,m)$ is the set of all functions from $\ff_2^n$ to $\ff_2^m$ and $Z_n$ is a vector 
    formed of $n$ independent copies of $Z$.
\end{definition}
Paper \cite{yu2019simulation} established the following result, stated here for a particular case of Bernoulli sources.
\begin{proposition}\cite[Theorem 10]{yu2019simulation}
    Let $Z$ be a Bernoulli($\delta$) random variable. Then 
    \begin{align*}
        I_p(Z) = 
    \begin{cases}
       h_p(\delta) & \text{ for } {p} \in \{0\} \cup [1,\infty]\\[.1in]
    h(\delta)& \text{ for } {p} \in (0,1),
    \end{cases}
    \end{align*}
where  $h_p(\delta)=\frac1{1-p}\log_2(\delta^p+(1-\delta)^p)$ is the two-point R{\' e}nyi entropy and $h=h_1$.
\end{proposition}

\begin{remark} Because of the equivalence stated in Proposition \ref{prop:equiv}, if we replace $D_p$ with $l_p$ distance or $l_p$-smoothness in Definition \ref{def: int_ran}, we still obtain $I_p(Z) = h_p(\delta)$ for $p \in (1,\infty)$.  
\end{remark}

The construction of the uniform distribution in \cite{yu2019simulation} involves rearranging the 
masses of the Bernoulli distribution into $2^{n(1- R)}$ bins having approximately equal 
probabilities, which is a computationally involved procedure. 
In this section, we show that RM codes are capable of extracting randomness from Bernoulli sources (cf. Lemma~\ref{lem:sm_to_proj}). The following theorem is a consequence of \cite[Theorem 6]{pathegama2023smoothing}.  

\begin{theorem}\label{thm: RM_intrin} 
    Let $R \in (0,1)$ and let $\cC_n$ be a sequence of RM codes whose rate $R_n$ approaches $R$. Let $\hH_n$ be the parity check matrix of $\cC_n$ and let $Z_n$ be a binary vector formed of independent Bernoulli$(\delta)$ random bits. If  $R > 1-h_{p}(\delta)$, then
\begin{align*}    
    &\lim_{n \to \infty} D_{p}(P_{\hH_n Z_n}\|P_{U_{n(1-R_n)}}) = 0, \quad {p} \in \{2,\dots,\infty\}
    \end{align*}
If $p=1$ and $R > (1-2\delta)^2$, then
\begin{align*}
    \lim_{n \to \infty} D(P_{\hH_n Z_n}\|P_{U_{n(1-R_n)}}) = 0.
\end{align*}
\end{theorem}

\begin{proof}
From Lemma \ref{lem:sm_to_proj}, we know that for all ${p} \in (0,\infty]$, 
   $$
   \|2^{n(1-R_n)} P_{\hH_n Z_n} \|_{p} = \|2^{n} P_{X_{\cC_n} + Z_n} \|_{p}.
   $$ 
Equivalence of the $l_p$-smoothness and $D_{p}$-smoothness \eqref{eq: Renyi smoothness} implies that
    $$
    D_{p}(P_{\hH_n Z_n}\|P_{U_{n(1-R_n)}}) = D_{p}(P_{X_{\cC_n} + Z_n}\|P_{U_n}).
    $$ 
Therefore, it suffices to show that for ${p} \in \{2,\dots,\infty\}$ and $R>1-h_p(\delta)$,
\begin{align*}
    \lim_{n \to \infty}D_{p}(P_{X_{\cC_n} + Z_n}\|P_{U_n}) = 0,
\end{align*} 
and that for $R>(1-2\delta)^2$,
\begin{align*}
    \lim_{n \to \infty}D(P_{X_{\cC_n} + Z_n}\|P_{U_n}) = 0,
\end{align*} 
In \cite[Theorem 6]{pathegama2023smoothing} we proved exactly this result for a general class of codes that includes RM codes. 
\end{proof}

Setting $\hat{R}_n=1-R_n$ in Theorem \ref{thm: RM_intrin}, we observe that for $\hat{R}:= 1-R> h_\alpha(\delta)$ and for $p \in \{2,3,\dots,\infty\},$ 
\begin{align*}
    \lim_{n \to \infty} D_{p}(P_{\hH_n Z_n}\|P_{U_{n\hat{R}_n}}) = 0,
\end{align*}
i.e., the intrinsic randomness is achievable by setting $f_n(Z_n) = \hH_n Z_n$.
To summarize, we observe that RM codes yield a computationally efficient alternative for distilling a string of (nearly) uniform bits from a Bernoulli source at rate close to the intrinsic randomness $I_p(Z)$ of the source.

\section{Concluding remarks}
The obtained results suggest the following open questions, already mentioned in the main text.
It is of interest to attempt better bounds for the symbol loss 
than those mentioned in Remark~\ref{rm: loss}, which may result in better bounds for $l_\infty$-smoothness.
Another question is to derive similar results for ensembles of linear codes smaller than the ensemble of all linear codes of fixed dimension considered here. Finally, it would be interesting to show that RM codes achieve intrinsic randomness of Bernoulli sources $I_p(Z)$ for $p=1$. Per 
the results of \cite{pathegama2023smoothing}, this is equivalent to
showing that nested sequences of RM codes achieve secrecy capacity of the binary wiretap 
channel in the strong sense.


\appendix

\subsection{Proof of Lemma \ref{lemma: balanced}}

Let us begin with the following definition that allows us to use certain symmetry properties of random linear codes.
\begin{definition}\label{def:balanced}
Let ${p}$ be a positive integer and let $\sC$ be a family of $[n,k]_q$ linear codes in $\ff_q^n$. 
We call $\sC$ a ${p}$-balanced code family if for all $d \in [0,{p}]$, any $[{p},d]$-tuple $(v_1,\dots,v_{{p}})$  appears in 
the same number of codes from $\sC$.
\end{definition}
Observe that the set of all $[n,k]_q$ linear codes forms a ${p}$-balanced collection for all 
integer $p\ge 1$. For $p=1$ this definition appears routinely in proofs of the GV bound. 

Let $T_{d}(n,{p}) = T_{d}^{(q)}(n,{p})$ be the number of  $[{p},d]$-tuples in  $\mathbbm{F}_q^n$.  If $\cC$ is a $k$-dimensional linear code in $\mathbbm{F}_q^n$, the number of $[{p},d]$-tuples formed by codewords in $\cC$ is given by $T_{d}^{(q)}(k,{p})$. 
With these definitions, we proceed to the following lemma.

\begin{lemma}\label{lemma: balanced 1}
Let $\sC$ be a ${p}$-balanced collection of codes in $\mathbbm{F}_q^n$. For any function $f : \mathbbm{F}_q^{n\times{p}} \to \mathbbm{R}$,
    \begin{multline}\label{eq: balanced_11}
        \frac{1}{|\sC|}\sum_{\cC\in \sC}\sum_{z_1,\dots,z_p\in\cC}f(z_1,\dots, z_{{p}})
        \\=  \sum_{d=0}^{{p}}\frac{T_{d}(k,{p})}{T_{d}(n,{p})}\sum_{(v_1^p\text{ is }[p,d])}f(v_1,\dots, v_{{p}}).
    \end{multline} 
\end{lemma}

\begin{proof} Suppose that each $d$-dimensional tuple $(v_1,\dots,v_{{p}})$ appears $S_d$ times in the 
family $\sC$ for some number $S_d\ge 1.$ 
Then we find that 
    \begin{multline}\label{eq: balanced_2}
        \sum_{\cC\in \sC}\sum_{z_1,\dots,z_p\in\cC}f(z_1,\dots, z_{{p}})\\
        =  \sum_{d=0}^{{p}}S_d\sum_{(v_1, \dots, v_{p})\text{ is }[p.d]}f(v_1,\dots, v_{{p}}).
    \end{multline}
    Letting $f(v_1,\dots,v_{p}) = \1{\{ \rank[v_1,\dots,v_{p}]=d\}}$, we obtain
    \begin{align*}
        |\sC|T_{d}(k,{p}) = S_d T_{d}(n,{p}).
    \end{align*}
Using this in \eqref{eq: balanced_2} completes the proof.
\end{proof}

Now let us estimate $T_d(n,{p})/T_d(k,{p})$. Observe that $T_d^{(q)}(n,{p})$ can be interpreted as the number of $n \times {p}$ matrices (over $\ff_q$) whose rank is exactly $d$.

It is well known that for $d \in [0,\min\{n,{p}\}]$ 
\begin{align*}
    T_{d}(n,{p}) = \prod_{j=0}^{d-1} \frac{(q^n-q^j)(q^{p}-q^j)}{q^d-q^j}.
\end{align*}
With this, we have 
\begin{align*}
    \frac{T_{d}(k,{p})}{T_{d}(n,{p})} 
    &=  \frac{\prod_{j=0}^{d-1}(q^k-q^j)}{\prod_{j=0}^{d-1}(q^n-q^j)}\\
    &= \prod_{j=0}^{d-1} \frac{q^{k-j}-1}{q^{n-j}-1}\\
    & \leq q^{d(k-n)}.
\end{align*}
Finally, since $ T_{d}(k,{p}) = 0,$ for $d > \min\{k,{p}\}$, the proof is complete.

\subsection{Proof of Lemma  \ref{lemma: norm_bound}}
As the first step of proving Lemma \ref{lemma: norm_bound}, we will establish an auxiliary 
estimate. In the proof, we will need the classical {\em rearrangement inequality} \cite[Sec. 10.2]{hardy1952inequalities}: given two sequences of nonnegative numbers $(a_1,\dots,a_n)$ and $(b_1,\dots,b_n)$, we would like to permute them to maximize the sum $\sum_{i}a_{\tau(i)}b_{\sigma(i)}$. 
The inequality states that the maximum is attained when both sequences are arranged in a
monotone (nondecreasing or nonincreasing) order. This statement extends to multiple sequences, see e.g., \cite{ruderman1952two}.
\begin{lemma}\label{lemma: norm_bound 1}
 Let $d\in \{1,\dots,p-1\}$ be fixed and let $v_j\in \ff_q^n, j\in[d].$ Choose arbitrary $p-d$ vectors $x_i\in \ff_q^d\backslash\{0\}, i\in[d+1,p]$ and construct vectors $v_{d+1},\dots,v_p$ as linear combinations of
the vectors $v_i, i\in[d]$ as follows:
   \begin{equation}\label{eq:lc}
   v_i=\sum_{j=1}^d x_{i}(j)v_j, \quad i=d+1,\dots,p.
   \end{equation}
Then for a function $f: \mathbbm{F}_q^n \to [0,\infty)$ 
    \begin{align*}
        \frac{1}{q^{nd}}\sum_{v_1,\dots,v_d \in \mathbbm{F}_q^n} f(v_1)f(v_2)\dots f(v_{p}) \leq \|f\|_1^{d-1}\|f\|_{{p}-d+1}^{{p}-d+1}.
    \end{align*}   
\end{lemma}
\begin{proof}
    Let us partition $[{p}]$ into $d$ classes as follows:
    \begin{align*}
        B_d =& \{i \in [d+1,{p}]: x_i(d) \neq 0\} \cup \{d\} \\
        B_j =& \{i \in [d+1,{p}]: x_i(j) \neq 0, i \notin \cap_{l=j+1}^{d}B_l\}\cup \{j\}, \\
        &\hspace*{2in}j\in[d-1].
    \end{align*}
In words, class $B_j$ contains $j$ and the indices of the vectors $v_i, i\in[p]$ whose expansion \eqref{eq:lc} includes $v_j$, but not $v_{j+1},\dots,v_d$. Note that some of the vectors $v_i, i\ge d+1$ may be simply copies of one of the vectors 
$v_i, i\in[d]$. Therefore,
    \begin{align*}
        \sum_{v_1,\dots,v_d \in \mathbbm{F}_q^n} f(v_1)f(v_2)\dots f(v_{p}) = \prod_{i=1}^d \sum_{v_i \in \mathbbm{F}_q^n}\prod_{l \in B_i}f(v_l).
    \end{align*}

Observe that when $v_i, i \in [d]$  runs through all the elements in $\mathbbm{F}_q^n$, while keeping the other vectors $v_j, j\in [i-1]$ fixed, the vector $v_l, l \in B_i$ also runs through all the elements in $\mathbbm{F}_q^n$. i.e., $f(v_i)$ and $f(v_l), l \in B_i$ as function evaluations of $v_i$ are permutations of each other. Applying the rearrangement inequality, we find that
     \begin{align*}
        \sum_{v_i \in \mathbbm{F}_q^n}\prod_{l \in B_i}f(v_l) \leq \sum_{v_i \in \mathbbm{F}_q^n} f(v_i)^{|B_i|}.
    \end{align*}
Denoting $b_i:=|B_i|$, we obtain
    \begin{align*}
        \frac{1}{q^{nd}}\sum_{v_1,\dots,v_d \in \mathbbm{F}_q^n} f(v_1)f(v_2)\dots f(v_{p}) \leq \prod_{i=1}^d \|f\|_{b_i}^{b_i}.
    \end{align*}

Now all that is left is to show that $\prod_{i=1}^d \|f\|_{b_i}^{b_i} \leq \|f\|_1^{d-1}\|f\|_{{p}-d+1}^{{p}-d+1}$. This is immediate from Lyapunov's inequality\footnote{The Lyapunov inequality \cite[Thm.17]{hardy1952inequalities} states that 
if $0<r<s<t$, then $\|f\|_s^s\le \|f\|_r^{r\frac{t-s}{t-r}}\|f\|_t^{t\frac{s-r}{t-r}}$.}, which implies that
\begin{align*}
         \|f\|_{b_i}^{b_i} \leq \|f\|_1^\frac{{p}-d+1-{b_i}}{{p}-d}\|f\|_{{p}-d+1}^\frac{({p}-d+1)({b_i}-1)}{{p}-d}.
\end{align*}
Therefore,
\begin{multline*}
    \prod_{i=1}^d \|f\|_{b_i}^{b_i} \leq \|f\|_1^{\frac{({p}-d+1)d-{p}}{{p}-d}}\|f\|_{{p}-d+1}^{\frac{({p}-d+1)({p}-d)}{{p}-d}}  
      \\  = \|f\|_1^{d-1}\|f\|_{{p}-d+1}^{{p}-d+1},
\end{multline*}
concluding the proof.
\end{proof}

Now let us prove the bound in Lemma \ref{lemma: norm_bound}, copied here
for readers' convenience: for $1 \leq d \leq {p}-1$ 
\begin{multline}
    \sum_{(v_1^p\text{ is }[p,d])}f(v_1)\dots f(v_{{p}})\\ \leq \binom{{p}}{d}\sum_{m= 0}^{{p} -d} \binom{{p}-d}{m}(q^d-1)^{{p}-d-m}q^{nd}\\ \times f(0)^m\|f\|_1^{d-1}\|f\|_{{p}-d-m+1}^{{p}-d-m+1}. \label{eq:lemma}
\end{multline}

\begin{proof} (of \eqref{eq:lemma})
Any $[{p},d]$-tuple can be written as a permutation of ${p}$ vectors, of which the first $d$ 
are linearly independent and the remaining ones are their linear combinations.
Let $\cV_d:=\{v_1,\dots,v_d\}$ and let $\langle\cV_d\rangle$ be its linear span.
We have
    \begin{align*}
        \sum_{(v_1^p\text{ is }[p,d])}&f(v_1)\dots f(v_{{p}}) 
        \\ &\leq \binom{{p}}{d}\sum_{(v_1^p\text{ is }[p,d])}f(v_1)\dots f(v_d)
        \\ &\hspace*{.8in} \times\sum_{v_{d+1},\dots, v_{{p}}\in \langle\cV_d\rangle}f(v_{d+1})\dots f(v_{p})
        \\
        & \leq \binom{{p}}{d}\sum_{v_1,\dots,v_d \in \mathbbm{F}_q^n}f(v_1)\dots f(v_d)
        \\ &\hspace*{.8in} \times\sum_{v_{d+1},\dots, v_{{p}}\in \langle\cV_d\rangle}f(v_{d+1})\dots f(v_{p}).
    \end{align*}
    Note that $(v_i)_{d+1}^{p}$ are linear combinations of $(v_i)_{i}^{d}$. Let $G =[v_1,\dots,v_d]$. Then for $j \in [d+1,p]$, we can write $v_j = Gx_j$ for some $x_j \in  \mathbbm{F}_q^d$. This implies
    \begin{multline*}
        \sum_{v_{d+1},\dots, v_{{p}}\in \langle\cV_d\rangle}f(v_{d+1})\dots f(v_{p}) 
        \\
        = \sum_{x_{d+1},\dots, x_{{p}}  \in \mathbbm{F}_q^d}f(Gx_{d+1})\dots f(Gx_{p}).
    \end{multline*}
Therefore,
    \begin{align*}
         \sum_{v_1,\dots,v_d } &f(v_1)\dots f(v_d) \sum_{v_{d+1},\dots, v_{{p}}\in \langle\cV_d\rangle}f(v_{d+1})\dots f(v_{p})\\
         =& \sum_{v_1,\dots,v_d }f(v_1)\dots f(v_d) \\
         &\hspace*{.8in} \times\sum_{x_{d+1},\dots, x_{{p}}  \in \mathbbm{F}_q^d}f(Gx_{d+1})\dots f(Gx_{p})\\
        =&\sum_{x_{d+1},\dots, x_{{p}}  \in \mathbbm{F}_q^d}\sum_{v_1,\dots,v_d }f(v_1)\dots f(v_d) \\ &\hspace*{.8in} \times f(Gx_{d+1})\dots f(Gx_{p}).
    \end{align*}
Now consider the inner sum $\sum_{v_1,\dots,v_d }f(v_1)\dots f(v_d)  f(Gx_{d+1})\dots f(Gx_{p})$, for a fixed choice of $x_{d+1},\dots, x_{{p}}$, where $m$ variables in $\{x_{d+1},\dots, x_{{p}}\}$ are set to $0$. Without loss of generality, we may assume that $x_{i} = 0$ for $i \in [d+1,d+m]$. In this case, 
\begin{align*}
     \sum_{v_1,\dots,v_d }&f(v_1)\dots f(v_d)  f(Gx_{d+1})\dots f(Gx_{p})\\
    = & f(0)^m \sum_{v_1,\dots,v_d }f(v_1)\dots f(v_d)  f(Gx_{d+m+1})\dots f(Gx_{p})\\
    \leq & q^{nd}f(0)^m\|f\|_1^{d-m-1}\|f\|_{{p}-d+1}^{{p}-d+1} \quad \text{(from Lemma \ref{lemma: norm_bound 1})}.
\end{align*}
Hence, 
    \begin{align*}
         \sum_{v_1,\dots,v_d }&f(v_1)\dots f(v_d) 
         \sum_{\substack{v_{d+1},\dots, v_{{p}}\\  \in span(v_1,\dots,v_d)}}f(v_{d+1})\dots f(v_{p})\\
        =&\sum_{x_{d+1},\dots, x_{{p}}  \in \mathbbm{F}_q^d}\sum_{v_1,\dots,v_d }f(v_1)\dots f(v_d)
        \\ &\hspace*{.8in} \times f(Gx_{d+1})\dots f(Gx_{p})\\
        \leq & \sum_{m= 0}^{{p} -d} \binom{{p}-d}{m}(q^d-1)^{{p}-d-m}q^{nd}
        \\ &\hspace*{.8in} \times f(0)^m\|f\|_1^{d-1}\|f\|_{{p}-d-m+1}^{{p}-d-m+1},
    \end{align*}
completing the proof.
\end{proof}

\vfill


\begin{thebibliography}{10}
\providecommand{\url}[1]{#1}
\csname url@samestyle\endcsname
\providecommand{\newblock}{\relax}
\providecommand{\bibinfo}[2]{#2}
\providecommand{\BIBentrySTDinterwordspacing}{\spaceskip=0pt\relax}
\providecommand{\BIBentryALTinterwordstretchfactor}{4}
\providecommand{\BIBentryALTinterwordspacing}{\spaceskip=\fontdimen2\font plus
\BIBentryALTinterwordstretchfactor\fontdimen3\font minus
  \fontdimen4\font\relax}
\providecommand{\BIBforeignlanguage}[2]{{%
\expandafter\ifx\csname l@#1\endcsname\relax
\typeout{** WARNING: IEEEtranS.bst: No hyphenation pattern has been}%
\typeout{** loaded for the language `#1'. Using the pattern for}%
\typeout{** the default language instead.}%
\else
\language=\csname l@#1\endcsname
\fi
#2}}
\providecommand{\BIBdecl}{\relax}
\BIBdecl

\bibitem{abbe2023reed}
E.~Abbe, O.~Sberlo, A.~Shpilka, and M.~Ye, ``{Reed-Muller} codes,''
  \emph{Foundations and Trends{\textregistered} in Communications and
  Information Theory}, vol.~20, no. 1--2, pp. 1--156, 2023.

\bibitem{alon1999linear}
N.~Alon, M.~Dietzfelbinger, P.~B. Miltersen, E.~Petrank, and G.~Tardos,
  ``Linear hash functions,'' \emph{Journal of the ACM (JACM)}, vol.~46, no.~5,
  pp. 667--683, 1999.

\bibitem{arimoto1973converse}
S.~Arimoto, ``On the converse to the coding theorem for discrete memoryless
  channels,'' \emph{IEEE Transactions on Information Theory}, vol.~19, no.~3,
  pp. 357--359, 1973.

\bibitem{bennett1995generalized}
C.~H. Bennett, G.~Brassard, C.~Cr{\'e}peau, and U.~M. Maurer, ``Generalized
  privacy amplification,'' \emph{IEEE Transactions on Information Theory},
  vol.~41, no.~6, pp. 1915--1923, 1995.

\bibitem{bloch2013strong}
M.~R. Bloch and J.~N. Laneman, ``Strong secrecy from channel resolvability,''
  \emph{IEEE Transactions on Information Theory}, vol.~59, no.~12, pp.
  8077--8098, 2013.

\bibitem{brakerski2019worst}
Z.~Brakerski, V.~Lyubashevsky, V.~Vaikuntanathan, and D.~Wichs, ``Worst-case
  hardness for {LPN} and cryptographic hashing via code smoothing,'' in
  \emph{Annual International Conference on the Theory and Applications of
  Cryptographic Techniques}.\hskip 1em plus 0.5em minus 0.4em\relax Springer,
  2019, pp. 619--635.

\bibitem{cachin1997entropy}
C.~Cachin, ``Entropy measures and unconditional security in cryptography,''
  Ph.D. dissertation, ETH Zurich, 1997.

\bibitem{canetti2006mitigating}
R.~Canetti, S.~Halevi, and M.~Steiner, ``Mitigating dictionary attacks on
  password-protected local storage,'' in \emph{Advances in Cryptology-CRYPTO
  2006: 26th Annual International Cryptology Conference, Santa Barbara,
  California, USA, August 20-24, 2006. Proceedings 26}.\hskip 1em plus 0.5em
  minus 0.4em\relax Springer, 2006, pp. 160--179.

\bibitem{carter1977universal}
J.~L. Carter and M.~N. Wegman, ``Universal classes of hash functions,'' in
  \emph{Proceedings of the Ninth Annual ACM Symposium on Theory of Computing},
  1977, pp. 106--112.

\bibitem{chou2015polar}
R.~A. Chou, M.~R. Bloch, and E.~Abbe, ``Polar coding for secret-key
  generation,'' \emph{IEEE Transactions on Information Theory}, vol.~61,
  no.~11, pp. 6213--6237, 2015.

\bibitem{cover2007capacity}
T.~M. Cover and H.~H. Permuter, ``Capacity of coordinated actions,'' in
  \emph{2007 IEEE International Symposium on Information Theory}.\hskip 1em
  plus 0.5em minus 0.4em\relax IEEE, 2007, pp. 2701--2705.

\bibitem{debris2023smoothing}
T.~Debris-Alazard, L.~Ducas, N.~Resch, and J.-P. Tillich, ``Smoothing codes and
  lattices: Systematic study and new bounds,'' \emph{IEEE Transactions on
  Information Theory}, vol.~69, no.~9, pp. 6006--6027, 2023.

\bibitem{debris2022worst}
T.~Debris-Alazard and N.~Resch, ``Worst and average case hardness of decoding
  via smoothing bounds,'' \emph{Cryptology ePrint Archive}, 2022, paper
  2022/1744.

\bibitem{dhar2022linear}
M.~Dhar and Z.~Dvir, ``Linear hashing with $l_\infty$ guarantees and two-sided
  {K}akeya bounds,'' in \emph{2022 IEEE 63rd Annual Symposium on Foundations of
  Computer Science (FOCS)}.\hskip 1em plus 0.5em minus 0.4em\relax IEEE, 2022,
  pp. 419--428.

\bibitem{dodis2004possibility}
Y.~Dodis, S.~J. Ong, M.~Prabhakaran, and A.~Sahai, ``On the (im)possibility of
  cryptography with imperfect randomness,'' in \emph{45th Annual IEEE Symposium
  on Foundations of Computer Science}.\hskip 1em plus 0.5em minus 0.4em\relax
  IEEE, 2004, pp. 196--205.

\bibitem{fehr2008randomness}
S.~Fehr and C.~Schaffner, ``Randomness extraction via $\delta$-biased masking
  in the presence of a quantum attacker,'' in \emph{Theory of Cryptography
  Conference}.\hskip 1em plus 0.5em minus 0.4em\relax Springer, 2008, pp.
  465--481.

\bibitem{guruswami2022punctured}
V.~Guruswami and J.~Mosheiff, ``Punctured low-bias codes behave like random
  linear codes,'' in \emph{2022 IEEE 63rd Annual Symposium on Foundations of
  Computer Science (FOCS)}.\hskip 1em plus 0.5em minus 0.4em\relax IEEE, 2022,
  pp. 36--45.

\bibitem{han1993approximation}
T.~S. Han and S.~Verd{\'u}, ``Approximation theory of output statistics,''
  \emph{IEEE Transactions on Information Theory}, vol.~39, no.~3, pp. 752--772,
  1993.

\bibitem{hardy1952inequalities}
G.~H. Hardy, J.~E. Littlewood, and G.~P{\'o}lya, \emph{Inequalities}.\hskip 1em
  plus 0.5em minus 0.4em\relax Cambridge University Press, 1952.

\bibitem{haastad1999pseudorandom}
J.~H{\aa}stad, R.~Impagliazzo, L.~A. Levin, and M.~Luby, ``A pseudorandom
  generator from any one-way function,'' \emph{SIAM Journal on Computing},
  vol.~28, no.~4, pp. 1364--1396, 1999.

\bibitem{hayashi2006general}
M.~Hayashi, ``General nonasymptotic and asymptotic formulas in channel
  resolvability and identification capacity and their application to the
  wiretap channel,'' \emph{IEEE Transactions on Information Theory}, vol.~52,
  no.~4, pp. 1562--1575, 2006.

\bibitem{hayashi2016equivocations}
M.~Hayashi and V.~Y. Tan, ``Equivocations, exponents, and second-order coding
  rates under various {R}{\'e}nyi information measures,'' \emph{IEEE
  Transactions on Information Theory}, vol.~63, no.~2, pp. 975--1005, 2017.

\bibitem{hkazla2021codes}
J.~H{\k a}z\l{}a, A.~Samorodnitsky, and O.~Sberlo, ``On codes decoding a
  constant fraction of errors on the {BSC},'' in \emph{Proceedings of the 53rd
  Annual ACM SIGACT Symposium on Theory of Computing}, 2021, pp. 1479--1488.

\bibitem{impagliazzo1989pseudo}
R.~Impagliazzo, L.~A. Levin, and M.~Luby, ``Pseudo-random generation from
  one-way functions,'' in \emph{Proceedings of the Twenty-First Annual ACM
  Symposium on Theory of Computing}, 1989, pp. 12--24.

\bibitem{karger1993global}
D.~R. Karger, ``Global min-cuts in {RNC}, and other ramifications of a simple
  min-cut algorithm.'' in \emph{Soda}, vol.~93.\hskip 1em plus 0.5em minus
  0.4em\relax Citeseer, 1993, pp. 21--30.

\bibitem{kaslasi2021public}
I.~Kaslasi, R.~D. Rothblum, and P.~N. Vasudevanr, ``Public-coin statistical
  zero-knowledge batch verification against malicious verifiers,'' in
  \emph{Annual International Conference on the Theory and Applications of
  Cryptographic Techniques}.\hskip 1em plus 0.5em minus 0.4em\relax Springer,
  2021, pp. 219--246.

\bibitem{luby1985simple}
M.~Luby, ``A simple parallel algorithm for the maximal independent set
  problem,'' in \emph{Proceedings of the seventeenth annual ACM symposium on
  Theory of computing}, 1985, pp. 1--10.

\bibitem{luzzi2023optimal}
L.~Luzzi, C.~Ling, and M.~R. Bloch, ``Optimal rate-limited secret key
  generation from {G}aussian sources using lattices,'' \emph{IEEE Transactions
  on Information Theory}, vol.~69, no.~8, pp. 4944--4960, 2023.

\bibitem{metropolis1953equation}
N.~Metropolis, A.~W. Rosenbluth, M.~N. Rosenbluth, A.~H. Teller, and E.~Teller,
  ``Equation of state calculations by fast computing machines,'' \emph{The
  journal of chemical physics}, vol.~21, no.~6, pp. 1087--1092, 1953.

\bibitem{micciancio2007worst}
D.~Micciancio and O.~Regev, ``Worst-case to average-case reductions based on
  {G}aussian measures,'' \emph{SIAM Journal on Computing}, vol.~37, no.~1, pp.
  267--302, 2007.

\bibitem{miller1975riemann}
G.~L. Miller, ``Riemann's hypothesis and tests for primality,'' in
  \emph{Proceedings of the Seventh Annual ACM Symposium on Theory of
  Computing}, 1975, pp. 234--239.

\bibitem{mosheiff2020ldpc}
J.~Mosheiff, N.~Resch, N.~Ron-Zewi, S.~Silas, and M.~Wootters, ``{LDPC} codes
  achieve list decoding capacity,'' in \emph{2020 IEEE 61st Annual Symposium on
  Foundations of Computer Science (FOCS)}.\hskip 1em plus 0.5em minus
  0.4em\relax IEEE, 2020, pp. 458--469.

\bibitem{nisan1996extracting}
N.~Nisan, ``Extracting randomness: how and why. {A} survey,'' \emph{Proceedings
  of Computational Complexity (Formerly Structure in Complexity Theory)}, pp.
  44--58, 1996.

\bibitem{pathegama2023smoothing}
M.~Pathegama and A.~Barg, ``Smoothing of binary codes, uniform distributions,
  and applications,'' \emph{Entropy}, vol.~25, no.~11, p. 1515, 2023.

\bibitem{polyanskiy2010arimoto}
Y.~Polyanskiy and S.~Verd{\'u}, ``Arimoto channel coding converse and
  {R}{\'e}nyi divergence,'' in \emph{2010 48th Annual Allerton Conference on
  Communication, Control, and Computing}.\hskip 1em plus 0.5em minus
  0.4em\relax IEEE, 2010, pp. 1327--1333.

\bibitem{rabin1980probabilistic}
M.~O. Rabin, ``Probabilistic algorithm for testing primality,'' \emph{Journal
  of Nnumber Theory}, vol.~12, no.~1, pp. 128--138, 1980.

\bibitem{rao2024criterion}
A.~Rao and O.~Sprumont, ``A criterion for decoding on the binary symmetric
  channel,'' \emph{Advances in Mathematics of Communications}, vol.~19, no.~2,
  pp. 437--477, 2025.

\bibitem{rennie1969stirling}
B.~C. Rennie and A.~J. Dobson, ``On {S}tirling numbers of the second kind,''
  \emph{Journal of Combinatorial Theory}, vol.~7, no.~2, pp. 116--121, 1969.

\bibitem{rivest1983randomized}
R.~L. Rivest and A.~T. Sherman, ``Randomized encryption techniques,'' in
  \emph{Advances in Cryptology: Proceedings of Crypto 82}.\hskip 1em plus 0.5em
  minus 0.4em\relax Springer, 1983, pp. 145--163.

\bibitem{ruderman1952two}
H.~D. Ruderman, ``Two new inequalities,'' \emph{The American Mathematical
  Monthly}, vol.~59, no.~1, pp. 29--32, 1952.

\bibitem{samorodnitsky2016entropy}
A.~Samorodnitsky, ``On the entropy of a noisy function,'' \emph{IEEE
  Transactions on Information Theory}, vol.~62, no.~10, pp. 5446--5464, 2016.

\bibitem{samorodnitsky2019upper}
------, ``An upper bound on $\ell_{q} $ norms of noisy functions,'' \emph{IEEE
  Transactions on Information Theory}, vol.~66, no.~2, pp. 742--748, 2019.

\bibitem{Simon2015}
B.~Simon, \emph{Real Analysis: {A} Comprehensive Course in Analysis, {P}art
  1}.\hskip 1em plus 0.5em minus 0.4em\relax American Mathematical Society,
  2015.

\bibitem{skorski2015shannon}
M.~Sk{\'o}rski, ``Shannon entropy versus {R}{\'e}nyi entropy from a
  cryptographic viewpoint,'' in \emph{IMA International Conference on
  Cryptography and Coding}.\hskip 1em plus 0.5em minus 0.4em\relax Springer,
  2015, pp. 257--274.

\bibitem{tan2018analysis}
V.~Y. Tan and M.~Hayashi, ``Analysis of remaining uncertainties and exponents
  under various conditional {R}{\'e}nyi entropies,'' \emph{IEEE Transactions on
  Information Theory}, vol.~64, no.~5, pp. 3734--3755, 2018.

\bibitem{tyagi2023information}
H.~Tyagi and S.~Watanabe, \emph{Information-Theoretic Cryptography}.\hskip 1em
  plus 0.5em minus 0.4em\relax Cambridge University Press, 2023.

\bibitem{Vadhan2012Pseudo}
S.~P. Vadhan, ``Pseudorandomness,'' \emph{Foundations and Trends in Theoretical
  Computer Science}, vol.~7, no. 1--3, pp. 1--336, 2012.

\bibitem{vembu1995generating}
S.~Vembu and S.~Verd{\'u}, ``Generating random bits from an arbitrary source:
  Fundamental limits,'' \emph{IEEE Transactions on Information Theory},
  vol.~41, no.~5, pp. 1322--1332, 1995.

\bibitem{vitter1985random}
J.~S. Vitter, ``Random sampling with a reservoir,'' \emph{ACM Transactions on
  Mathematical Software (TOMS)}, vol.~11, no.~1, pp. 37--57, 1985.

\bibitem{waters2009dual}
B.~Waters, ``Dual system encryption: {R}ealizing fully secure {IBE} and {HIBE}
  under simple assumptions,'' in \emph{Annual International Cryptology
  Conference}.\hskip 1em plus 0.5em minus 0.4em\relax Springer, 2009, pp.
  619--636.

\bibitem{yu2019simulation}
L.~Yu and V.~Y. Tan, ``Simulation of random variables under {R}{\'e}nyi
  divergence measures of all orders,'' \emph{IEEE Transactions on Information
  Theory}, vol.~65, no.~6, pp. 3349--3383, 2019.

\end{thebibliography}
\end{document}